\documentclass[draftclsnofoot,onecolumn,10pt]{IEEEtran}
\pdfoutput=1
\makeatletter

\usepackage{setspace}
\usepackage{subfig}
\usepackage{epsfig}  % for figures
\usepackage{graphicx}  % another package that works for figures
\usepackage{amsmath}  % for math spacing
\usepackage{amssymb}  % for math spacing
\usepackage{url}  % Hyphenation of URLs.
\usepackage{lscape}  % Useful for wide tables or figures.
\usepackage[justification=raggedright]{caption}	% makes captions ragged right - thanks to Bryce Lobdell
\usepackage{bbm}
\usepackage{stfloats}
\usepackage{float}
\usepackage{units}
\usepackage[boxed,scleft]{algorithm2e}
\usepackage{enumerate}
\usepackage{comment}
\usepackage{flushend}
\usepackage{import}
\usepackage{color}

\newtheorem{theorem}{Theorem}[section]
\newtheorem{lemma}[theorem]{Lemma}

\newtheorem{prop}{Proposition}[section]
\newtheorem{proof}{Proof}

\graphicspath{{./figures/}}
\DeclareGraphicsExtensions{.pdf,.png,.jpg}

\DeclareMathOperator*{\diag}{diag}

\DeclareMathOperator*{\SNR}{SNR}
\DeclareMathOperator*{\avg}{avg}

\DeclareMathOperator*{\var}{Var}

\begin{document}

\title{MIMO Communications over Multi-Mode Optical Fibers: Capacity Analysis and Input-Output Coupling Schemes}
%\author{Peter Kairouz and Eric Kim}

\author{\IEEEauthorblockN{Peter Kairouz and Andrew Singer}\\
 \IEEEauthorblockA{Coordinated Science Laboratory / Department
of Electrical and Computer Engineering\\
 University of Illinois at Urbana Champaign, Urbana, Illinois,
61801\\
 Email: \{kairouz2, acsinger\}@illinois.edu} }

\maketitle

%\tableofcontents
%\listoffigures

\begin{abstract}
We consider multi-input multi-output (MIMO) communications over multi-mode
fibers (MMFs). Current MMF standards, such as OM3 and OM4, use fibers with
core radii of 50\,$\mu$m, allowing hundreds of modes to propagate.
Unfortunately, due to physical and computational complexity limitations, we
cannot couple and detect hundreds of data streams into and out of the
fiber. In order to circumvent this issue, two solutions were presented in
the literature. The first is to design new fibers with smaller radii so
that they can support a desired number of modes. The second is to design
multi-core fibers with a reasonable number of cores. However, both
approaches are expensive as they necessitate the replacement of currently
installed fibers. Moreover, both approaches have limited future
scalability. In our work, we present input-output coupling schemes that
allow the user to couple and extract a reasonable number of signals from a
fiber with many modes (such as those used in OM3 and OM4 standards). This
approach is particularly attractive as it is scalable; i.e., the fibers do
not have to be replaced every time the number of transmitters or receivers
is increased, a phenomenon that is likely to happen in the near future. In
addition, fibers with large radii can support higher peak powers relative
to fibers with small radii while still operating in the linear regime.
However, the only concern is that fibers with more modes suffer from
increased mode-dependent losses (MDLs). Our work addresses this last
concern.

We present a statistical channel model that incorporates intermodal
dispersion, chromatic dispersion, mode dependent losses, and mode coupling.
We show that the statistics of the fiber's frequency response are
independent of frequency. This simplifies the computation of the average
Shannon capacity of the fiber. We later extend this model to include input
and output couplers and provide an input-output coupling strategy that
leads to an increase in the overall capacity. This strategy can be used
whenever channel state information (CSI) is available at the transmitter
and the designer has full control over the couplers. We show that the
capacity of an $N_t \times N_t$ MIMO system over a fiber with $M \gg N_t$
modes can approach the capacity of an $N_t$-mode fiber with no
mode-dependent losses. Moreover, we present a statistical input-output
coupling model in order to quantify the loss in capacity when CSI is not
available at the transmitter or there is no control over the input-output
coupler. It turns out that the loss, relative to $N_t$-mode fibers, is
minimal (less than 0.5 dB) for a wide range of signal-to-noise ratios
(SNRs) and a reasonable range of MDLs. This means that there is no real
need to replace the already installed fibers and that our strategy is an
attractive approach to solving the above problem.
\end{abstract}

\IEEEpeerreviewmaketitle

\section{Introduction}
\label{ch:introduction}
Since Shannon defined the notion of channel capacity as the fundamental limit on achievable transmission rates with vanishing probability of error, system designers have attempted to reach this limit by leveraging device technology advances and increasingly sophisticated algorithms and architectures.  Moore's law, together with advances in signal processing, information theory, and coding theory have enabled us to essentially achieve this fundamental limit for a number of narrow-band wired and wireless communication links. Because of their superior bandwidth-distance product, optical fibers have become extremely popular and have largely replaced traditional copper wire technologies. Optical communication links can support serial data rates that are typically several orders of magnitude higher than their wired or wireless electrical counterparts, such as voice-band or cable modem technology or even high-speed chip-to-chip serial links. Despite their superiority, optical links have limited capacity and the circuits, signal processing, and information theory communities need to completely re-think the design and analysis of optical communication systems in order to address the ever increasing demand for Internet bandwidth. Multi-input multi-output (MIMO) communications over multi-mode fibers (MMFs) holds the promise of improving bandwidth efficiency. However, the capacity of MIMO optical links has not been fully investigated due to the lack of accurate and mathematically tractable channel models. In this paper, we present a detailed linear model for the MIMO multi-mode optical channel and analyze its capacity as a function of input-output coupling as well as other physical parameters. We also introduce an input-output coupling strategy and compare it to the uncontrolled coupling case in terms of achievable rate.

\subsection{Motivation}
In an information-intensive era, the demand for Internet bandwidth is increasing at a rate of $56\%$ per year, while the increase in supply is falling behind at a rate of $25\%$ per year \cite{Tkach2012}. The increase in demand is fueled by the boom in web-based data services such as cloud computing and real-time multimedia applications. As a result, optical fiber communication researchers are looking into new ways of boosting the transmission rate of optical links. Given that polarization division multiplexing (PDM) and wavelength division multiplexing (WDM) have already been exploited \cite{Winzer2009}, the only remaining degree of freedom is space division multiplexing \cite{Winzer2011outageprob_smots}. MIMO optical communication increases the transmission rates of MMF systems by multiplexing a number of independent data streams on different spatial modes. Note that, unlike WDM systems, all the laser sources in this case have the same wavelength. MMF is a dominant type of fiber used for high speed data communication in short-range links such as local area networks (LAN) and data centers \cite{Benner2005mmf_servers}. It is usually favored over single-mode fibers because of its relaxed connector alignment tolerances and its reduced transceiver connector costs. Plastic optical fibers are great examples of MMFs with remarkably low installation and operation costs \cite{Koike2008}. However, they suffer from mode-dependent losses, mode coupling, intermodal dispersion, and chromatic dispersion (group velocity dispersion) \cite{Agarwal2002fiber_optics}. All these phenomena will be explained in detail in Section \ref{prop_model}. These limitations make the design and analysis of MIMO multi-mode systems challenging yet exciting.

\subsection{Literature Review}
Following the work of Shannon \cite{Shannon1949mathematical}, many information theorists investigated the capacity of different channels, including single-input single-output (SISO) channels with memory, channels with constrained input alphabet, and multiple-input multiple-output (MIMO) channels. In their seminal works, Telatar and Foschini \emph{et al.}, independently showed that the capacity of a MIMO flat fading wireless channel, under the Raleigh fading model, scales linearly with respect to the minimum number of antennas at the transmitter and receiver \cite{telatar1999_capacity,Foschini1998wireless_limits}. Since then, the wireless communications community has been focused on developing detection and coding schemes for MIMO systems in order to achieve the aforementioned capacity gains. Recent wireless technologies such as WLAN 802.11n and Long Term Evolution Advanced (LTE-A) are examples of MIMO systems deploying up to 8 transmitters and receivers. More importantly, this MIMO technique is not limited to wireless systems.

Loosely speaking, the number of degrees of freedom (DoF) of a channel is an upper limit on the number of independent data streams that can be transmitted through the channel over a period of time. A more rigorous definition of DoF is given in \cite{tse_viswanath_wireless_2005} as the minimial dimension of the received signal space. The quadrature and in-phase components of a passband information signal are two familiar and commonly exploited degrees of freedom in wired and wireless communication systems. Frequency, time, code, quadrature, and polarization states are all well explored and already utilized in commercial optical systems. However, the spatial degree of freedom, which is unique to MMFs, has not been exploited yet in commercial products and is still under research. In 2000, H. R. Stuart was the first to notice the similarity between the multipath wireless channel and the MMF optical channel and suggested using the spatial modes to multiplex several independent data streams onto the fiber \cite{Stuart2000mimo_optical}. Prior to this finding, single-mode fibers were always considered to be superior to MMFs because of their improved bandwidth-distance product (as SMFs do not suffer from intermodal dispersion). However, we will show in Section \ref{ch:capacity} that MMFs have advantages over single-mode fibers from an information theoretic capacity perspective. Therefore, MIMO over MMF seems to be a better route to higher data rates. In fact, Stuart was the first to demonstrate the feasibility of a $2\times2$ MMF system and to show that there are indeed some capacity gains to be leveraged \cite{Stuart2000mimo_optical}. However, Stuart's analysis and experiments assumed a radio frequency sub-carrier ($\sim$ 1 GHz) instead of an optical carrier ($\sim$ 100 THz). This assumption was later relaxed in the work of Shah \emph{et al.} but their treatment did not account for any intermodal dispersion, chromatic dispersion, or mode coupling \cite{Shah2005_comimo}. Recently, the information theoretic capacity of coherent MMF systems has been studied in \cite{Winzer2011outageprob_smots}, where the authors ignored the frequency selectivity of the channel but incorporated the effects of mode coupling. In \cite{Keang-Po2011mdl_capacity}, Keang-Po \emph{et al.} considered the capacity of a frequency selective MMF channel at a particular frequency. They later studied the impact of frequency diversity on the channel capacity for mutli-mode fibers with 10 modes \cite{Keang-Po2011frequency_diversity}. However, their models did not incorporate the effect of mode-dependent phase shifts or chromatic dispersion.

\subsection{Outline and Contributions}
In Section \ref{ch:fundamentals}, we present a MIMO channel propagation model that takes intermodal dispersion, chromatic dispersion, mode-dependent losses, and mode coupling into account. In Section \ref{ch:capacity}, we compute the Shannon capacity of an $M$-mode fiber and demonstrate how mode-dependent losses and mode coupling affect it. In Section \ref{ch:coupling}, we analyze the coupling of a reasonable number of laser sources to a fiber with hundreds of modes. We also propose an input-output coupling model and present a coupling strategy: using the input-output couplers to perform a particular type of beamforming. This strategy allows the effective transmission of data along the least lossy subset of end-to-end eigenmodes. The resultant capacity is almost equal to that of a fiber with $N_t$ modes and no modal losses, an ideal case which maximizes the capacity of an $N_t \times N_t$ MIMO system. This coupling strategy can only be used when channel state information (CSI) is available at the transmitter and there is full control over the input-output couplers. In the absence of these conditions, an appropriate random input-output coupling model is used in order to better model the behavior of the system and quantify the expected loss in the fiber's capacity. It turns out that the loss, relative to $N_t$-mode fibers, is minimal (less than 0.5 dB) for a wide range of SNRs and a reasonable range of MDLs.
\subsection{Random Unitary Matrices}
\label{random_unitary_matrices}
In this section, we provide a brief overview on random unitary matrices and discuss the isotropic invariance property that will prove useful when we compute the capacity of the fiber in Section \ref{ch:capacity}. We define $\mathbb{U}\left(M\right):= \{\mathbf{U} \in \mathbb{C}^{M \times
M} | \mathbf{U^{*}U=UU^{*}=I_M}\}$ to be the space of $M\times M$ unitary
matrices. If the distribution of an $M\times N$ random matrix is invariant to
left (right) multiplication by any $M\times M$ ($N\times N$) deterministic
unitary matrix, it is called left (right) rotationally invariant. Assume the
probability distribution function (pdf) $f\left(\mathbf{A}\right)$ of an
$M\times N$ random matrix $\mathbf{A}$ exists, $\mathbf{A}$ is left
rotationally invariant if
\begin{equation}
f_{\mathbf{UA}}\left(\mathbf{UA}\right)=f\left(\mathbf{A}\right)
\end{equation}
where $\mathbf{U} \in \mathbb{U}\left(M\right)$, and $\mathbf{A}$ is right
rotationally invariant if
\begin{equation}
f_{\mathbf{AV}}\left(\mathbf{AV}\right)=f\left(\mathbf{A}\right)
\end{equation}
where $\mathbf{V} \in \mathbb{U}\left(N\right)$. A random matrix is
isotropically invariant if it is left and right rotationally invariant. It
turns out that $\mathbb{U}\left(M\right)$ forms a compact topological group,
and thus a unique uniform measure (up to a scalar multiplication), called
Haar measure, can be defined over $\mathbb{U}\left(M\right)$
\cite{Mezzadri2007_randomMatrix,Petz2004_asymptotics}. Random unitary
matrices of size $M\times M$ are random matrices sampled uniformly from
$\mathbb{U}\left(M\right)$.
\begin{lemma}
\label{isotropic_unitary} The pdf of $\mathbf{A}$, an $M\times M$ random
unitary matrix, is isotropically invariant \cite{Petz2004_asymptotics}.
\end{lemma}
This lemma will be helpful in the following section.
%\subsection{Notation}
%\begin{itemize}
%\item $x\left(t\right)$ is a time domain signal
%\item $x\left(\omega\right)$ is the continuous time Fourier transform (CTFT) of $x\left(t\right)$
%\item $\mathbf{x}$ $\left(\mathbf{X}\right)$ is a vector (matrix) of scalar entries
%\item $\mathbf{x}\left(t\right)$ $\left(\mathbf{X}\left(t\right)\right)$ is a vector (matrix) of continuous time signals
%\item $\mathbf{x}\left(\omega\right)$ $\left(\mathbf{X}\left(\omega\right)\right)$ is the entry-wise CTFT of $\mathbf{x}\left(t\right)$ $\left(\mathbf{X}\left(t\right)\right)$
%\item $\mathbf{I_N}$ denotes the $N \times N$ identity matrix
%\item $\mathbf{0_{L \times  K}}$ represents the $L \times  K$ zero matrix
%\item $x \star y\left(t\right)$ represents the convolution of $x\left(t\right)$ with $y\left(t\right)$
%\item $\mathbf{X^{*}}$ is the conjugate transpose of $\mathbf{X}$
%\item $\mathbf{X}^{T}$ is the transpose of $\mathbf{X}$
%\item $\det\left(\mathbf{X}\right)$ and $\tr\left(\mathbf{X}\right)$ denote the determinant and trace of $\mathbf{X}$ respectively
%\item $\diag\left(\mathbf{x}\right)$ represents a diagonal matrix formed by the entries of $\mathbf{x}$
%\item $f^{'}\left(x\right)$ is the first order derivative of $f\left(x\right)$
%\item $f^{''}\left(x\right)$ is the second order derivative of $f\left(x\right)$
%\end{itemize}

\section{Fundamentals and Modeling}
\label{ch:fundamentals}
Coherent systems use well calibrated phase controlled laser sources and local oscillators operating well in the terahertz regime (hundreds of THz) to transmit and recover the phase and amplitude of an information-bearing signal. On the other hand, non-coherent systems use simple LEDs and photo detectors to transmit and detect the energy of an information signal. Thus, coherent systems are more complex, more expensive, and harder to build and maintain when compared to non-coherent systems. This is why the majority of currently deployed optical systems are non-coherent while a small percentage of the high end systems are coherent. However, coherent systems are becoming more popular as the optoelectronic devices are becoming more affordable. In fact, the state-of-the-art optical systems use both polarization and quadrature multiplexing to multiply the data rate by a factor of four. For example, OC-768 systems use dual polarizations in addition to quadrature phase-shift keying (QPSK) to multiplex four independent data streams and transmit them all at the same time. The OC-768 network has transmission speeds of 40 Gbit/s. This means that in a system using QPSK and dual polarization, the transmitter operates at a frequency of about 10 GHz. Because coherent systems are becoming more popular and affordable, the capacity analysis we perform in Section \ref{ch:capacity} is exclusively applicable to coherent systems.

Electromagnetic waves propagating inside the core of a fiber are characterized by Maxwell's equations. When the core radius is sufficiently small, only one solution to the wave equations is supported and the fiber is said to be a single-mode fiber. In multi-mode fiber systems, the core radius is relatively large and hence there is more than one solution (propagation mode) to the wave equation \cite{Agarwal2002fiber_optics}. Ideally, the field inside the core would propagate in different orthogonal modes that do no interact with one another. However, due to manufacturing non-idealities and index of refraction inhomogeneities, the modes may couple. This phenomenon is called mode coupling and is modeled in Section \ref{prop_model}.
\subsection{Fiber Propagation Model}
\label{prop_model}
For coherent optical systems operating in the linear regime, the basic form of the baseband transfer function governing the input-output relationship of the $i^{th}$ mode is given by
\begin{equation}\label{fiberprop}H_{i}\left(x,y,z,\omega\right)=\tilde{\phi}_{i}\left(x,y,\omega\right)e^{-\frac{\kappa_i z}{2}}e^{-j\beta_i\left(\omega+\omega_c\right)z}
\end{equation}
where $\omega_c$ is the laser's center frequency, $\tilde{\phi}_{i}\left(x,y,\omega\right)$ is the transverse function (spatial pattern) of the $i^{th}$ mode, $\kappa_i$ is the mode-dependent attenuation factor, and $\beta_i\left(\omega+\omega_c\right)$ is the $i^{th}$ mode's propagation constant \cite{papen2012_lightwave}. Expanding the function $\beta_i\left(\omega\right)$ around $\omega_c$ using its Taylor series expansion, and keeping the first and second order derivative terms, we get
\begin{equation}
\label{aprrox_prop}
H_{i}\left(x,y,\omega\right)\approx \phi_{i}\left(x,y\right)e^{\frac{g_i}{2}}e^{-j\theta_i}e^{-j\omega\tau_i}e^{-j\omega^2\alpha_i}
\end{equation}
where $\phi_{i}\left(x,y\right)=\tilde{\phi}_{i}\left(x,y,L,\omega_c\right)$, $g_i=-\kappa_i L$, $\theta_i=\beta_i\left(\omega_c\right)L$, $\tau_i=\beta_i^{'}\left(\omega_c\right)L$, and $\alpha_i=\beta_i^{''}\left(\omega_c\right)L$. Observe that $z$ has been suppressed as it has been evaluated at $L$, the fiber's length. The function $\tilde{\phi}_{i}\left(x,y,\omega\right)$ generally depends on $\omega$ but since the signal spectrum (tens of GHz) is narrow around the laser's center frequency (hundreds of THz), we drop this dependency and evaluate it at $\omega_c$. The model in (\ref{aprrox_prop}) assumes that the propagation of the mode is completely characterized by a second order linear model where the only phenomena exhibited along the $i^{th}$ mode are
\begin{itemize}
\item \emph{mode-dependent loss} (MDL): $g_i=-\kappa_i L$
\item \emph{mode-dependent phase shift} (MDPS): $\theta_i=\beta_i\left(\omega_c\right)L$
\item \emph{group delay} (GD):  $\tau_i=\beta_i^{'}\left(\omega_c\right)L$
\item \emph{group velocity dispersion} (GVD): $\alpha_i=\beta_i^{''}\left(\omega_c\right)L$
\end{itemize}
The mode-dependent losses (MDLs) are negative quantities describing the attenuation experienced by the modal fields. On the other hand, the mode-dependent phase shifts (MDPSs) represent phase shifts experienced by the modal fields. In general, modal fields propagate at different speeds and thus the group delays (GDs) characterize the arrival times of different modes. Therefore, if we transmit a narrow pulse through the fiber, it would appear as a pulse having a width of $T_d=\max_{i,j}\{|\tau_i-\tau_j|\}$ at the output of the fiber. The quantity $T_d$ is referred to as the channel's delay spread. Assume, without loss of generality, that the group delays are sorted in increasing order, $\tau_1$ being the smallest and $\tau_M$ being the largest. In this case, $T_d$ is given by
\begin{eqnarray}
T_d&=&\max_{i,j}\{|\tau_i-\tau_j|\} \nonumber \\
&=&\tau_M-\tau_1 \nonumber \\
&=&L\left(\beta_M^{'}\left(\omega_c\right)-\beta_1^{'}\left(\omega_c\right)\right)
\end{eqnarray}
Thus, $T_d$ is directly proportional to the length of the fiber. The pulse broadening phenomenon, due to nonzero $T_d$, is called intermodal dispersion and is a serious performance limitation in MMF systems. The group velocity dispersion (GVD), also called chromatic dispersion (CD), suggests that different frequencies coupled to the same mode propagate at different speeds and hence broadening occurs to the field propagating in a particular mode. This phenomenon is called intra-modal dispersion. In a first-order model, intermodal dispersion is assumed to dominate over intra-modal dispersion and the GVD term is typically neglected, especially for shorter lengths $L$. Furthermore, since we are not interested in analyzing the field at every point $\left(x,y\right)$ in the fiber's core, we suppress this term to obtain the following expression:
\begin{equation}
\label{fiberpropaproxx}
H_{i}\left(\omega\right)\propto e^{\frac{g_i}{2}}e^{-j\theta_i}e^{-j\omega\tau_i}e^{-j\omega^2\alpha_i} \end{equation}
Ideally, the field at the output due to the $i^{th}$ mode is given by $r_i\left(t\right)=s_i\star h_i\left(t\right)$, where $s_i\left(t\right)$ is the field at the input due to the same mode. Thus, the frequency domain vector representation of the modal fields at the output of the fiber is given by
\begin{equation}
\left[
  \begin{array}{c}
    r_1\left(\omega\right) \\
    \vdots \\
    r_M\left(\omega\right) \\
  \end{array}
\right]
=\left[
                   \begin{array}{cccc}
                     H_1\left(\omega\right) &  &    \\
                      &    \ddots&  \\
                      &    &  H_M\left(\omega\right)\\
                   \end{array}
                 \right]\left[
  \begin{array}{c}
    s_1\left(\omega\right) \\
    \vdots \\
    s_M\left(\omega\right) \\
  \end{array}
\right]
\end{equation}
where the off-diagonal entries are zero because the modes are assumed to be orthogonal. This analysis neglects the existing fiber aberrations such as fiber bends, index of refraction inhomogeneities, and random vibrations, and is therefore incomplete. In fact, the modes interact with one another and exchange energy as they propagate along the fiber, complicating the analysis of the wave propagation. The treatment we present was first applied to polarization mode dispersion (PMD) in \cite{Khosravani2001pmd_coupling} and was then generalized to model mode coupling in \cite{Keang-Po2011mdl_capacity}. In the regime of high mode coupling, for example when plastic optical fibers are used, an MMF with $M$ modes\footnote{In this work, $M$ refers to all the available spatial degrees of freedom including the $x$ and $y$ polarization states.} is split into $K \gg 1$ statistically independent longitudinal sections as depicted in Figure \ref{separate_sections}.
\begin{figure}[t]
\centering
{
\includegraphics[scale=0.8]{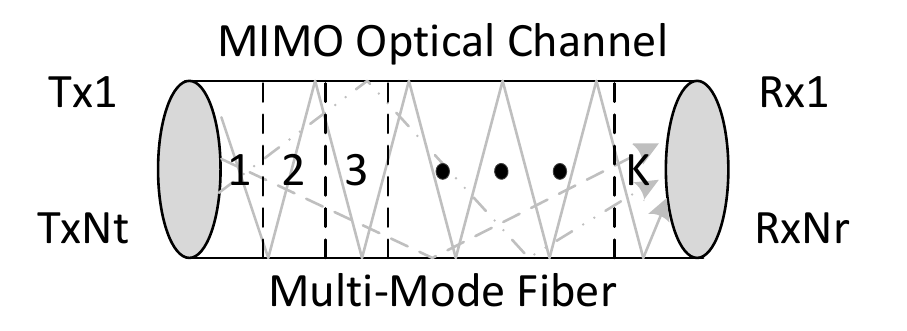}
\caption{A multi-mode fiber with $K$ propagation sections}
\label{separate_sections}
}
\end{figure}
The number of sections $K$ is equal to $L/l_c$, where $l_c$ represents the correlation length of the fiber. The frequency response of each section is given by
\begin{equation}
\mathbf{H}^{k}\left(\omega\right)=\mathbf{U}^{k}\mathbf{\Lambda}^{k}\left(\omega\right)\mathbf{V}^{k^{*}}~~~{\tt for}~~k=1,...,K
\end{equation}
where $\mathbf{U}^{k}$ and $\mathbf{V}^{k}$ are $M \times  M$ frequency-independent projection matrices (unitary matrices) describing the modal coupling via a phase and energy shuffling process at the input and output of each section and
\begin{equation}
\label{propagation_general}
\mathbf{\Lambda}^{k}\left(\omega\right)=\diag\left(e^{\frac{1}{2}g_1^{k}-j\theta_1^{k}-j\omega\tau_1^{k}-j\omega^2\alpha^{k}_1 },...,e^{\frac{1}{2}g_M^{k}-j\theta_M^{k}-j\omega\tau_M^{k}-j\omega^2\alpha^{k}_M }\right)
\end{equation}
is the propagation matrix describing the ideal (uncoupled) field propagation in the $k^{th}$ section. This model assumes that mode coupling occurs at the interface of different sections while the propagation in each section is ideal (and is described by $\mathbf{\Lambda}^{k}\left(\omega\right)$). In (\ref{propagation_general}), the vectors $\mathbf{g}^{k}=\left(g_1^{k}, ...,g_M^{k}\right)$, $\boldsymbol{\theta}^{k}=\left(\theta_1^{k}, ...,\theta_M^{k}\right)$, $\boldsymbol{\tau}^{k}=\left(\tau_1^{k}, ...,\tau_M^{k}\right)$, and $\boldsymbol{\alpha}^{k}=\left(\alpha_1^{k}, ...,\alpha_M^{k}\right)$ represent the uncoupled MDL, MDPS, GD, and GVD coefficients in the $k^{th}$ section. Here, $g_i^{k}=-\kappa_i^{k}l_c$, $\theta_i^{k}=\beta_i^{k}\left(\omega_c\right)l_c$, $\tau_i^{k}=\beta_i^{' k}\left(\omega_c\right)l_c$, and $\alpha^{k}_i=\beta_i^{''k}\left(\omega_c\right)l_c$ are not necessarily identical across the $M$ modes and $K$ sections and will be modeled as random variables in Section \ref{rand_prop_model}. The overall channel frequency response is equal to the product of the frequency responses of the $K$ sections and is given by
\begin{equation}
\label{fiber_response}
\mathbf{H}\left(\omega\right)=\mathbf{H}^{\left(K\right)}\left(\omega\right)...\mathbf{H}^{\left(1\right)}\left(\omega\right)
\end{equation}
Alternatively, one could describe the input output relationship in time domain by
\begin{equation}
\label{matrix_impulse_response}
\mathbf{H}\left(t\right)=\mathbf{H}^{\left(K\right)} \star \mathbf{H}^{\left(K-1\right)}...\mathbf{H}^{\left(2\right)} \star \mathbf{H}^{\left(1\right)}\left(t\right)
\end{equation}
In (\ref{matrix_impulse_response}), the operation $\mathbf{C}\left(t\right)=\mathbf{A}\star\mathbf{B}\left(t\right)$ represents a matrix convolution operation. Specifically, the $\left(i,j\right)^{th}$ entry of $\mathbf{C}\left(t\right)$ is given by
\begin{equation}
c_{ij}\left(t\right)=\sum_{l=1}^{M}a_{il}\star b_{lj}\left(t\right)
\end{equation}
where $M$ is the dimension of the square matrices $\mathbf{A}\left(t\right)$ and $\mathbf{B}\left(t\right)$ and $\star$ denotes the convolution operator.
\subsection{Random Propagation Model}
\label{rand_prop_model}
We now develop a random propagation model for the MIMO optical channel. The random model we introduce is an extended variant of what was presented in \cite{Keang-Po2011mdl_capacity} and \cite{Keang-Po2011gd_statistics}. The per-section coupling matrices $\mathbf{U}^{k}$ and $\mathbf{V}^{k}$ are modeled as independent and identically distributed (i.i.d.) random unitary matrices with arbitrary distributions. We assume that the propagation characteristics $g_i^{k}$, $\theta_i^{k}$, $\alpha_i^{k}$, and $\tau_i^{k}$ are all independent random quantities. In addition, each of $\mathbf{g}^{k}$, $\boldsymbol{\theta}^{k}$, $\boldsymbol{\tau}^{k}$, and $\boldsymbol{\alpha}^{k}$ has zero mean identically distributed, but possibility correlated, entries. The zero mean assumption is not restrictive because the mean MDL, MDPS, GD, and GVD do not affect the capacity of the fiber. Even though the propagation characteristics are identically distributed within a particular section, they need not have the same distributions from one section to the other. We define $\sigma_{k}$ to be the standard deviation of the uncoupled MDLs in the $k^{th}$ section: $\sigma_{k}=\sqrt{\var\left(g_i^{k}\right)}=l_c\sqrt{\var\left(\kappa_i^{k}\right)}$. At any fixed frequency $\omega_0$, the overall frequency response in (\ref{fiber_response}) can be written as
\begin{equation}
\label{fiber_responsefixed}
\mathbf{H}\left(\omega_0\right)=\mathbf{U}_{H}\left(\omega_0\right)\mathbf{\Lambda}_{H}\left(\omega_0\right)\mathbf{V}_{H}^{*}\left(\omega_0\right)
\end{equation}
by the singular value decomposition (SVD) of $\mathbf{H}\left(\omega_0\right)$. In (\ref{fiber_responsefixed}), all the matrices are random frequency dependent square matrices and
\begin{equation}
\mathbf{\Lambda}_{H}\left(\omega_0\right)=\diag\left(e^{\frac{1}{2}\rho_1},...,e^{\frac{1}{2}\rho_M}\right)
\end{equation}
contains the \emph{end-to-end eigenmodes}, singular values of $\mathbf{H}\left(\omega_0\right)$. We note that the end-to-end eigenmodes are not actual solutions to the wave equation, but rather they characterize the effective overall propagation through the fiber. The vector $\boldsymbol{\rho}=\left(\rho_1,\rho_2,...,\rho_M\right)$ contains the \emph{end-to-end mode-dependent losses}, the logarithms of the eigenvalues of $\mathbf{H\left(\omega_0\right)H^{*}\left(\omega_0\right)}$. These quantities are obviously frequency dependent random variables as they are the logarithms of the eigenvalues of a frequency dependent random matrix. The \emph{accumulated mode-dependent loss variance} is defined as
\begin{equation}
\label{accumulated_mdl}
\xi^2=\sigma_{1}^2+\sigma_{2}^2+...+\sigma_{K}^2
\end{equation}
where $\xi$ is measured in units of the logarithm of power gain and can be converted to decibels by multiplying its value by $10/\ln10$ \cite{Keang-Po2011mdl_capacity}. When all sections have identical distributions for the MDLs, Equation (\ref{accumulated_mdl}) reduces to $\xi^2=K\sigma^2$ because $\sigma_{k}=\sigma$ for all $k$.

\section{Capacity of Multi-Mode Fibers}
\label{ch:capacity}
In this section, we compute the capacity of coherent MMF systems under the presence of mode-dependent phase shifts (MDPSs), mode-dependent losses (MDLs), group delay (GD), chromatic dispersion (CD), and mode coupling. Table \ref{tab:mysecondtable} summarizes the parameters governing the random propagation model presented in Section \ref{rand_prop_model}.
\begin{table}[!htb]
\centering
\caption{Random propagation model}
\label{tab:mysecondtable}
\begin{tabular}{|c|c|}
  \hline
fiber's frequency response & $\mathbf{H}\left(\omega\right)=\mathbf{H}^{K}\left(\omega\right)...\mathbf{H}^{1}\left(\omega\right)$ \\
   \hline
per-section response &   $\mathbf{H}^{k}\left(\omega\right)=\mathbf{U}^{k}\mathbf{\Lambda}^{k}\left(\omega\right)\mathbf{V}^{k^{*}}$ \\
\hline
per-section coupling matrices & $\mathbf{U}^k$ and $\mathbf{V}^k$ \\
\hline
uncoupled MDL & $\mathbf{g}^{k}=\left(g_1^{k}, ...,g_M^{k}\right)$ \\
\hline
uncoupled MDPS & $\boldsymbol{\theta}^{k}=\left(\theta_1^{k}, ...,\theta_M^{k}\right)$ \\
\hline
uncoupled GD & $\boldsymbol{\tau}^{k}=\left(\tau_1^{k}, ...,\tau_M^{k}\right)$ \\
\hline
uncoupled GVD & $\boldsymbol{\alpha}^{k}=\left(\alpha_1^{k}, ...,\alpha_M^{k}\right)$ \\
\hline
uncoupled MDL variance & $\sigma_{k}^2=\var\left(g_i^{k}\right)=l_c^2\var\left(\kappa_i^{k}\right)$ \\
\hline
accumulated MDL variance & $\xi^2=\sigma_{1}^2+\sigma_{2}^2+...+\sigma_{K}^2$ \\
\hline
\end{tabular}
\end{table}
Each of the vectors $\mathbf{g}^{k}$, $\boldsymbol{\theta}^{k}$, $\boldsymbol{\tau}^{k}$, and $\boldsymbol{\alpha}^{k}$ has zero mean identically distributed, but possibility correlated, entries. Moreover, the vectors $\mathbf{g}^{k_1}$, $\boldsymbol{\theta}^{k_1}$, $\boldsymbol{\tau}^{k_1}$, and $\boldsymbol{\alpha}^{k_1}$ are independent of $\mathbf{g}^{k_2}$, $\boldsymbol{\theta}^{k_2}$, $\boldsymbol{\tau}^{k_2}$, and $\boldsymbol{\alpha}^{k_2}$ for $k_1\ne k_2$. However, they can have the same statistical distributions. Recall, from Section \ref{rand_prop_model}, that the $k^{th}$ section propagation matrix is given by
\begin{eqnarray}
\mathbf{\Lambda}^{k}\left(\omega\right)&=&\diag\left(e^{\frac{1}{2}g_1^{k}-j\theta_1^{k}-j\omega\tau_1^{k}-j\omega^2\alpha^{k}_1 },...,e^{\frac{1}{2}g_M^{k}-j\theta_M^{k}-j\omega\tau_M^{k}-j\omega^2\alpha^{k}_M }\right) \nonumber \\
&=& \Theta^{k}\mathbf{T}^{k}\mathbf{A}^{k}\mathbf{G}^{k}
\end{eqnarray}
where $\Theta^{k}=\diag\left(e^{-j\theta_1^{k}},...,e^{-j\theta_M^{k}}\right)$, $\mathbf{T}^{k}=\diag\left(e^{-j\omega\tau_1^{k}},...,e^{-j\omega\tau_M^{k}}\right)$, $\mathbf{A}^{k}=\diag\left(e^{-j\omega^2\alpha^{k}_1},...,e^{-j\omega^2\alpha^{k}_M }\right)$, and $\mathbf{G}^{k}=\diag\left(e^{\frac{1}{2}g_1^{k}},...,e^{\frac{1}{2}g_M^{k}}\right)$.
\subsection{Frequency Flat Channel Capacity}
We first study the capacity of the system when the channel's delay spread and CD are negligible. The more general frequency selective case is handled in Section \ref{frequency_selective_capacity}. In this regime, $\max_{ij}|\tau^k_i-\tau^k_j|\approx0$ and $\max_{i}|\alpha^k_i|\approx0$ and hence $\tau^k_i=\tau^k$ and $\alpha^k_i=0$ for all $i$ and $k$. Therefore, the $k^{th}$ section propagation matrix is given by
\begin{eqnarray}
\mathbf{\Lambda}^{k}\left(\omega\right)&=&\diag\left(e^{\frac{1}{2}g_1^{k}-j\theta_1^{k}-j\omega\tau_1^{k}-j\omega^2\alpha^{k}_1 },...,e^{\frac{1}{2}g_M^{k}-j\theta_M^{k}-j\omega\tau_M^{k}-j\omega^2\alpha^{k}_M }\right) \nonumber \\
&=& e^{-j\omega\tau^{k}}\diag\left(e^{\frac{1}{2}g_1^{k}-j\theta_1^{k}},...,e^{\frac{1}{2}g_M^{k}-j\theta_M^{k}}\right) \nonumber \\
&=& e^{-j\omega\tau^{k}}\Theta^{k}\mathbf{G}^{k} \nonumber \\
&=& e^{-j\omega\tau^{k}}\mathbf{\Lambda}^{k}
\end{eqnarray}
where $\mathbf{\Lambda}^{k}=\Theta^{k}\mathbf{G}^{k}$. Therefore, the overall response can be written as
\begin{eqnarray}
\label{frequency_flat_model}
\mathbf{H}\left(\omega\right)&=&e^{-j\omega\sum_{k=1}^{K}\tau^{k}}\mathbf{U}^{K}\mathbf{\Lambda}^{K}\mathbf{V}^{K^{*}}...\mathbf{U}^{1}\mathbf{\Lambda}^{1}\mathbf{V}^{1^{*}} \nonumber \\
&=&e^{-j\omega\sum_{k=1}^{K}\tau^{k}}\mathbf{U}_{H}\mathbf{\Lambda}_{H}\mathbf{V}_{H}^{*}
\end{eqnarray}
where $\mathbf{U}_{H}$, $\mathbf{\Lambda}_{H}$, and $\mathbf{V}_{H}^{*}$ are obtained by applying the singular value decomposition to $\mathbf{U}^{K}\mathbf{\Lambda}^{K}\mathbf{V}^{K^{*}}...\mathbf{U}^{1}\mathbf{\Lambda}^{1}\mathbf{V}^{1^{*}}$. Observe that $\mathbf{U}_{H}$, $\mathbf{\Lambda}_{H}$, and $\mathbf{V}_{H}^{*}$ are all frequency independent. The term $e^{-j\omega\sum_{k=1}^{K}\tau^{k}}$ is a delay term and can be neglected if we assume that the transmitter and receiver are synchronized. Thus, the channel is frequency flat and is given by
\begin{equation}
\label{channel_svvd}
\mathbf{H}=\mathbf{U}_{H}\mathbf{\Lambda}_{H}\mathbf{V}_{H}^{*}
\end{equation}
Consequently, the input-output relationship under this frequency flat channel model in (\ref{channel_svvd}) is given by
\begin{equation}
\label{frequency_flat_model_2}
\mathbf{y}=\mathbf{H}\mathbf{x}+\mathbf{v}
\end{equation}
where $\mathbf{x}$ and $\mathbf{y}$ represent the transmitted and received vectors, respectively, and $\mathbf{v}$ represents the modal noise which is modeled as additive white Gaussian noise (AWGN) with covariance matrix $N_0\mathbf{I_M}$, $N_0$ being the noise power density per Hz. This assumes that coherent optical communication is used and that electronic noise is the dominant source of noise. In addition, the fiber non-linearities are neglected under the assumption that the signal's peak to average power ratio (PAPR) and peak power are both low enough. This condition is not restrictive because MMFs have large radii and hence can support more power (relative to single mode fibers) while still operating in the linear region. The input-output model in (\ref{frequency_flat_model_2}) may seem identical to the wireless MIMO flat fading one. However, the Rayleigh fading i.i.d. model does not hold in our case because $\mathbf{H}$ is a product of $K$ terms, each containing a random diagonal matrix sandwiched between two random unitary matrices. Moreover, the entries of $\mathbf{H}$ are correlated. From \cite{telatar1999_capacity}, the capacity of a single instantiation of the channel in (\ref{channel_svvd}), when channel state information (CSI) is not available at the transmitter, is given by
\begin{eqnarray}
\label{channel_normal_capacity}
C\left(\mathbf{H}\right)&=&\log \det\left(\mathbf{I_{M}}+\frac{\SNR}{M}\mathbf{HH^{*}}\right) \nonumber \\
&=&\sum_{n=1}^{M}\log\left(1+ \frac{\SNR}{M}\lambda_n^2\right)~~~{\tt b/s/Hz}
\end{eqnarray}
where $\SNR=P/N_0W$, $P$ representing the total power divided equally across all modes and $W$ representing the available bandwidth in Hz. The $\lambda_n^2$'s are the eigenvalues of $\mathbf{HH^{*}}$. If CSI is available at the transmitter, the capacity could be further increased through waterfilling \cite{Cover2006_infotheory,tse_viswanath_wireless_2005}. In this case, the transmitter pre-processes the transmit vector $\mathbf{x}$ by allocating powers using a waterfilling procedure and then multiplies $\mathbf{x}$ by $\mathbf{V}_{H}$. On the other side, the receiver multiplies the received vector $\mathbf{y}$ by $\mathbf{U}^{*}_{H}$. This effectively turns the MIMO channel into a set of parallel AWGN channels. In optical communications, the beamforming process assumes that the designer can couple the fields of different sources onto the fiber exactly as determined by $\mathbf{V}_{H}$. This procedure, though beneficial, is complicated as it necessitates the design of sophisticated reconfigurable mode-selective spatial filters using coherent spatial light modulators \cite{Alon2004_spatialmodulators,Chen2011_modeselective}.

In the above analysis, we considered the capacity of (\ref{frequency_flat_model_2}) for a given instantiation of $\mathbf{H}$. However, since $\mathbf{H}$ is random,  the channel capacity $C\left(\mathbf{H}\right)$ is a random variable. In the fast fading regime, the ergodic capacity, expected value of $C\left(\mathbf{H}\right)$, is desired as it dictates the fastest rate of transmission \cite{tse_viswanath_wireless_2005}. On the other hand, in the slow fading regime, the cumulative distribution function (CDF) of $C\left(\mathbf{H}\right)$ is desired as it determines the probability of an outage event for a particular rate of transmission \cite{tse_viswanath_wireless_2005}. In either case, the cumulative distribution and the expected value of $C\left(\mathbf{H}\right)$ are both functions of the distribution of $\boldsymbol{\lambda}=\left(\lambda_1^2,...,\lambda_M^2\right)$, the eigenvalues of $\mathbf{HH^{*}}$. From (\ref{channel_svvd}), the matrix $\mathbf{HH}^{*}=\mathbf{U}_{H}\mathbf{\Lambda}^2_{H}\mathbf{U}^{*}_{H}$ is Hermitian and its eigenvalues, the squares of the singular values of $\mathbf{H}$, are real non-negative quantities. Recall from Section \ref{rand_prop_model} that the quantity $\lambda_n=e^{\frac{1}{2}\rho_n}$ refers to the $n^{th}$ end-to-end eigenmode and the quantity $\rho_n$ refers to the $n^{th}$ end-to-end mode dependent loss (MDL). The distribution of the end-to-end MDL values was studied in \cite{Keang-Po2011mdl_capacity} where it was shown that as $M$ tends to infinity, the $\rho_n$'s become independent and identically distributed on a semicircle. Their analysis and simulations, however, did not incorporate the effect of mode dependent phase shifts (MDPSs), $\theta_i^k$'s. We now show that the statistical distribution of the end-to-end MDL values is unchanged even when MDPSs are incorporated.
\begin{theorem}
\label{unitary_property}
The statistics of $\mathbf{H}$ are independent of mode dependent phase shifts.
\end{theorem}
\begin{proof}
We show that the statistics of $\mathbf{H}^{k}=\mathbf{U}^{k}\Theta^{k}\mathbf{G}^{k}\mathbf{V}^{k^{*}}$ are the same as those of $\mathbf{H}^{k}=\mathbf{U}^{k}\mathbf{G}^{k}\mathbf{V}^{k^{*}}$ for all $k=1,...,K$. Observe that $\Theta^{k}$ is a unitary matrix that is also random because it has random orthonormal columns. However, $\Theta^{k}$ does not necessarily belong to the class of random unitary matrices as it is not necessarily uniformly distributed over $\mathbb{U}\left(M\right)$. Nonetheless, we note that the distribution of $\mathbf{W}=\mathbf{U}^{k}\Theta^{k}$ is the same as the distribution of $\mathbf{U}^{k}$ because
\begin{eqnarray}
f\left(\mathbf{W}\right)&=&\int_{\Theta^{k}}f\left(\mathbf{W}|\Theta^{k}\right)f\left(\Theta^{k}\right)d\Theta^{k} \nonumber \\
&=&f\left(\mathbf{U}^{k}\right)\int_{\Theta^{k}}f\left(\Theta^{k}\right)d\Theta^{k} \nonumber \\
&=&f\left(\mathbf{U}^{k}\right)
\end{eqnarray}
where the second equality holds because for a given instantiation of $\Theta^{k}$, the random matrix $\mathbf{W}|\Theta^{k}$ has the same distribution as $\mathbf{U}^{k}$ (by Lemma \ref{isotropic_unitary} in Section \ref{random_unitary_matrices}). Therefore, the statistics of $\mathbf{H}$ are unchanged when the MDPSs are incorporated, and thus the results in \cite{Keang-Po2011mdl_capacity} carry over to this more general setting.
\end{proof}
Figure \ref{mdl_distribution} shows that for $M=100$ the distribution of $\rho_n$ approaches a semicircle.
\begin{figure}
\centering
{
\subfloat[M=8]{\includegraphics[width=0.33\textwidth]{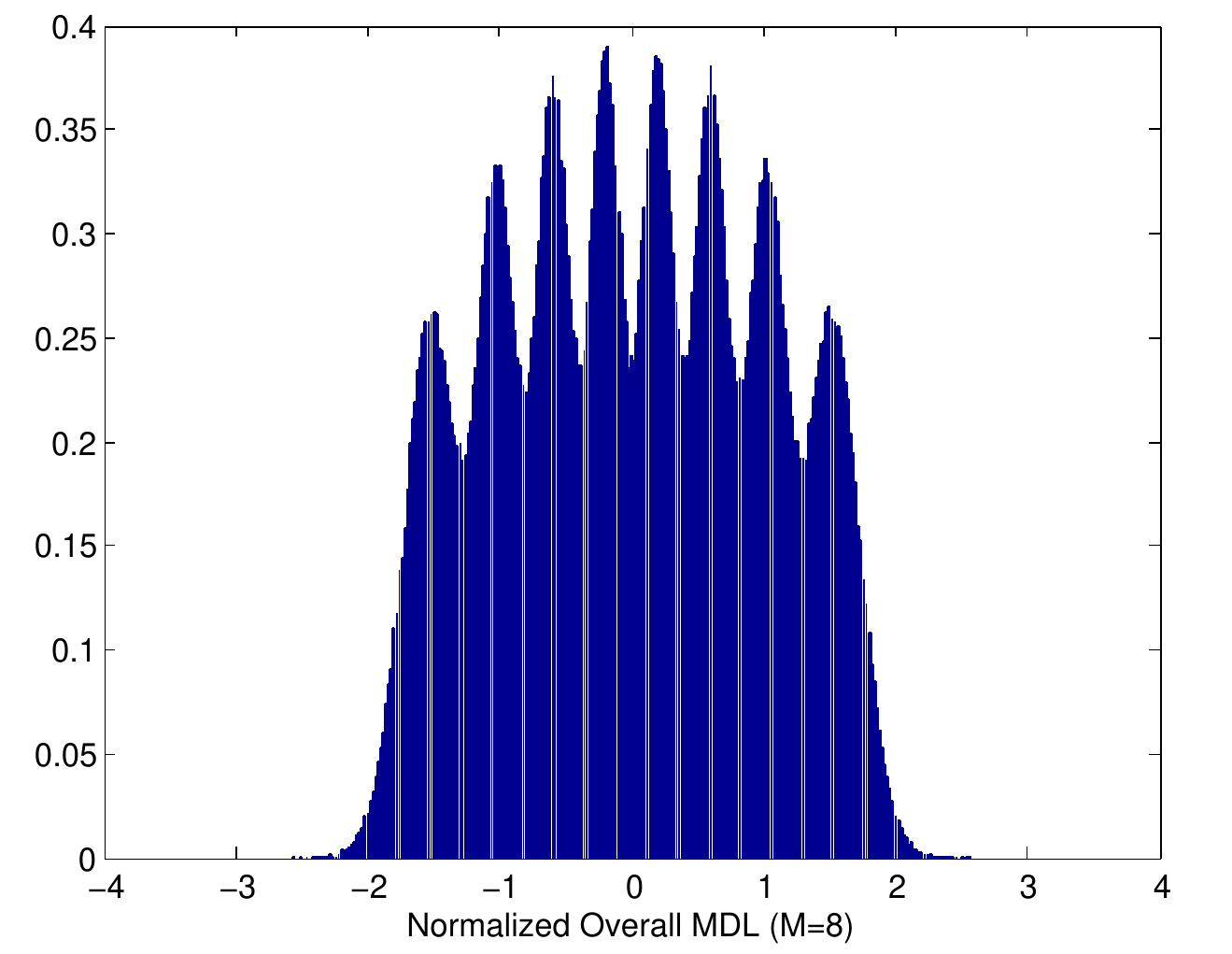}}
\subfloat[M=52]{\includegraphics[width=0.33\textwidth]{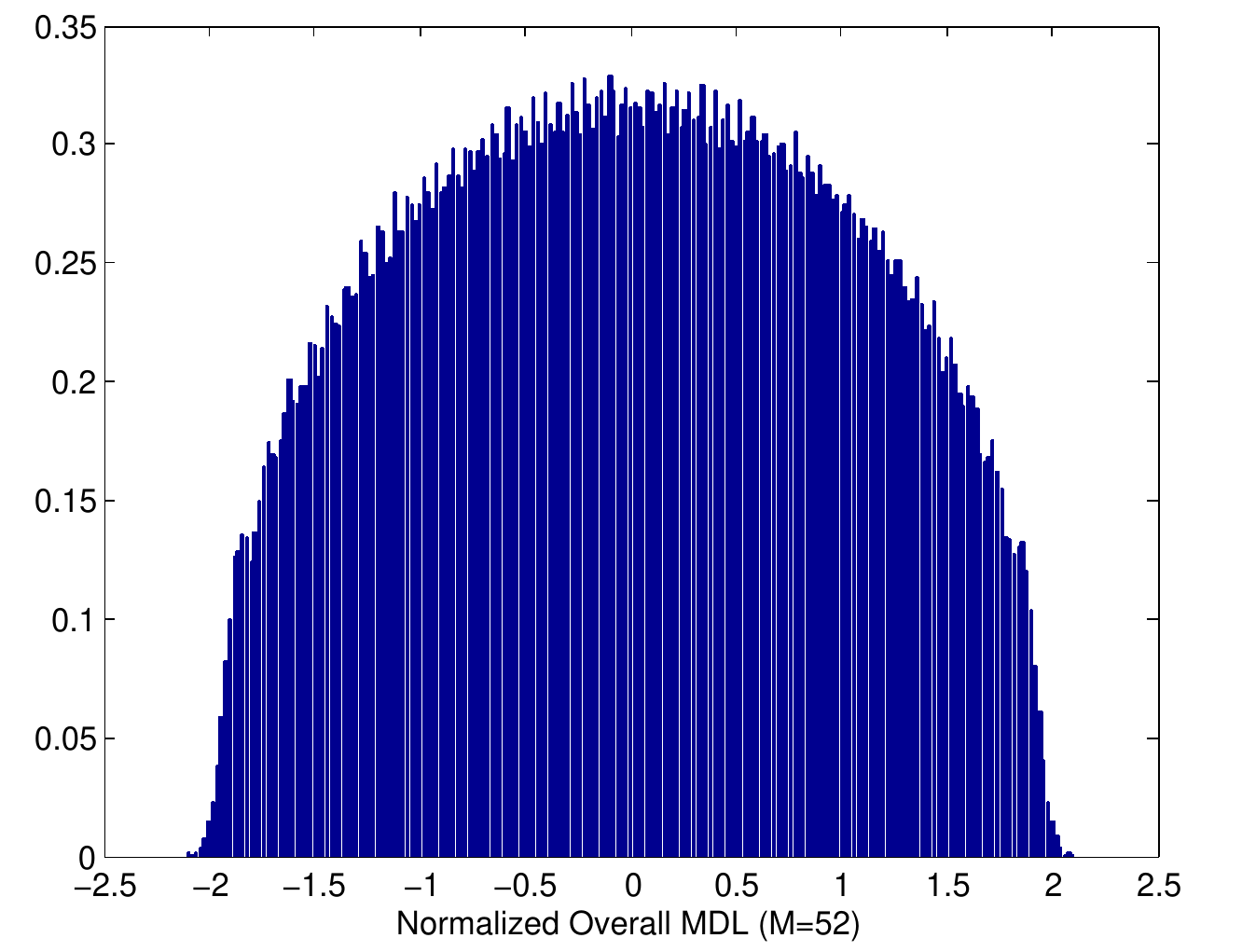}}
\subfloat[M=100]{\includegraphics[width=0.33\textwidth]{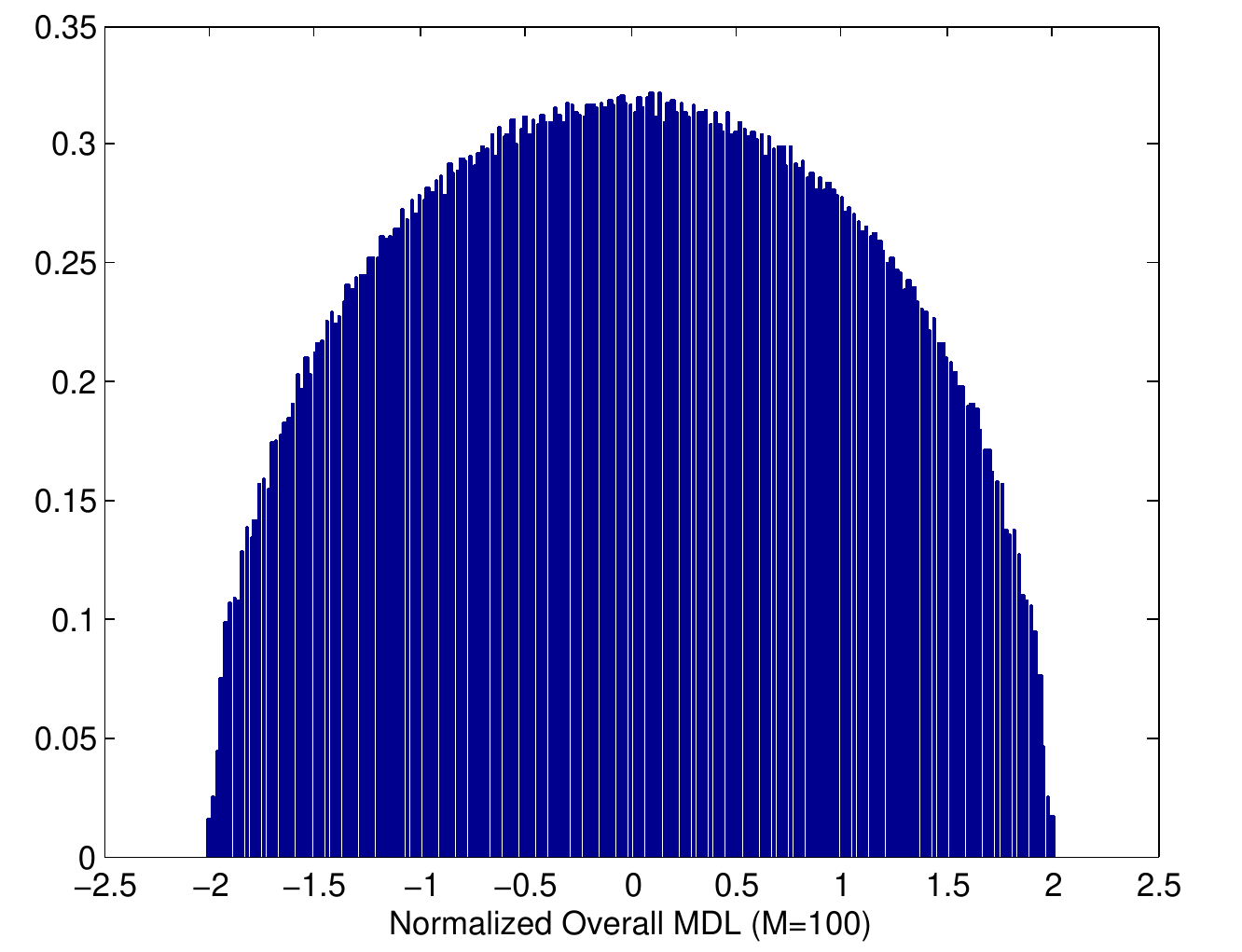}}
\caption{Distribution of end-to-end MDL}
\label{mdl_distribution}
}
\end{figure}
The distributions in Figure \ref{mdl_distribution} were obtained by generating a large sample of channel matrices $\mathbf{H}$ (for $M = 8$, $M = 52$, and $100$) and estimating the distributions of the logarithm of their singular values. Appendix \ref{appendix_random_matrices} explains how we can generate random unitary matrices, which are needed to create samples of $\mathbf{H}$, from matrices with i.i.d. complex Gaussian entries. In this section, we use the following notation
\begin{equation}
\tilde{x}=\frac{x}{\mathbb{E}\left[x\right]}
\end{equation}
where $\tilde{x}$ denotes the energy-normalized version of the random variable $x$. The average capacity of $C\left(\mathbf{H}\right)$ is given by
\begin{equation}
\label{average_capacity}
C_{\avg}=\sum_{n=1}^{M}\mathbb{E}\left[\log\left(1+ \frac{\SNR}{M}\tilde{\lambda}_n^2\right)\right]~~~{\tt b/s/Hz}
\end{equation}
where the average is taken over the statistics of the end-to-end MDL values \cite{tse_viswanath_wireless_2005}. Figure \ref{capacity_mdleffect1} shows the average capacity of MMFs for various values of $M$ and $\xi=4$ dB. The capacity of the system increases with an increasing number of modes. This is intuitive because as the number of modes increases, the fiber's spatial degrees of freedom are increased. Figure \ref{capacity_mdleffect2} shows the effect of accumulated MDLs on the average capacity. An increasing value of $\xi$ results in a capacity equivalent to that of a fiber with fewer modes. This means that as $\xi^2$, the accumulated mode-dependent loss variance, increases the system loses its spatial degrees of freedom.
\begin{figure}[H]
\centering
{
\subfloat[Capacity of MIMO MMF systems at $\xi=4$ dB]{\includegraphics[width=0.48\textwidth]{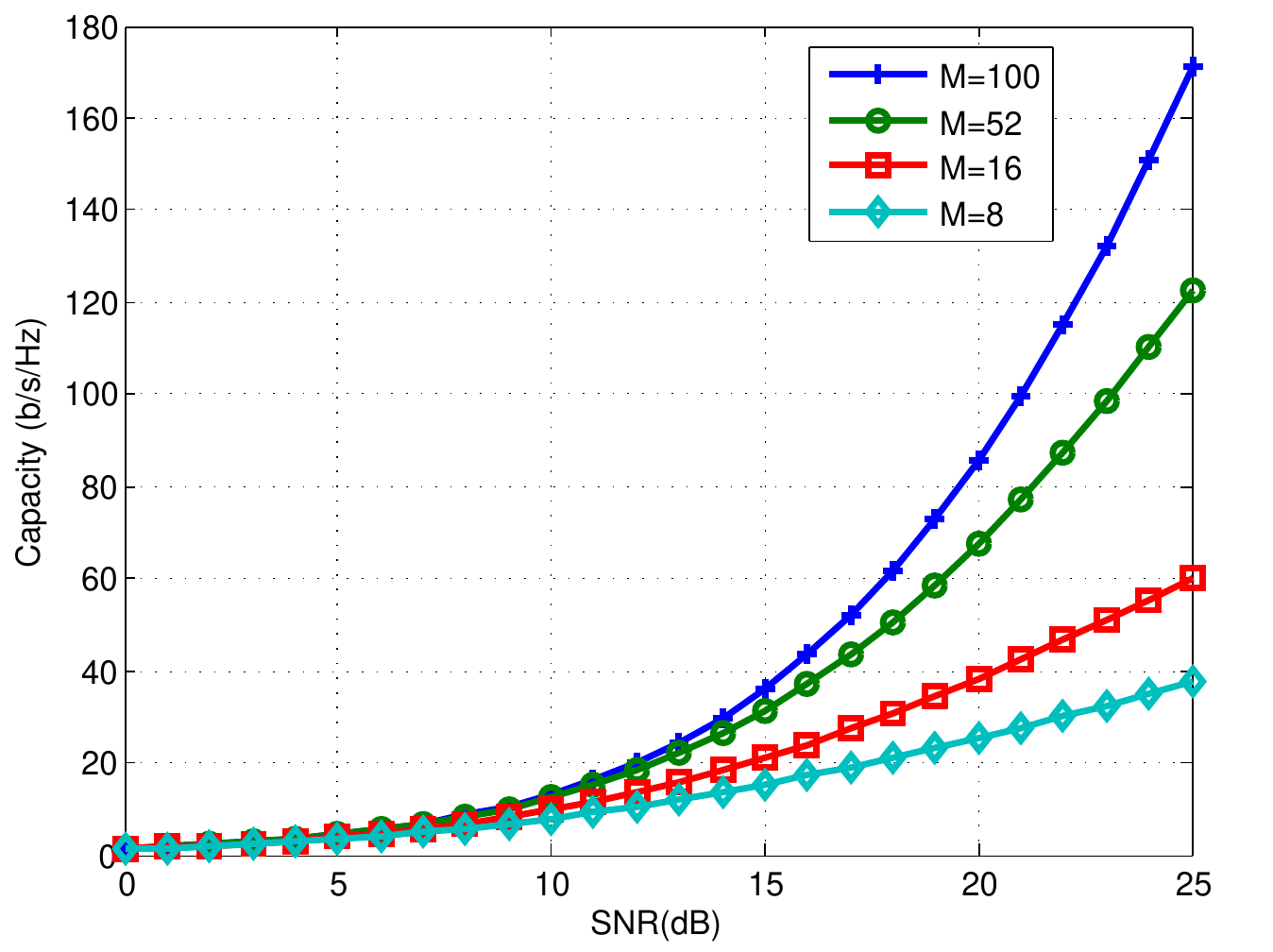}\label{capacity_mdleffect1}}
\subfloat[Effect of MDL for M=100]{\includegraphics[width=0.48\textwidth]{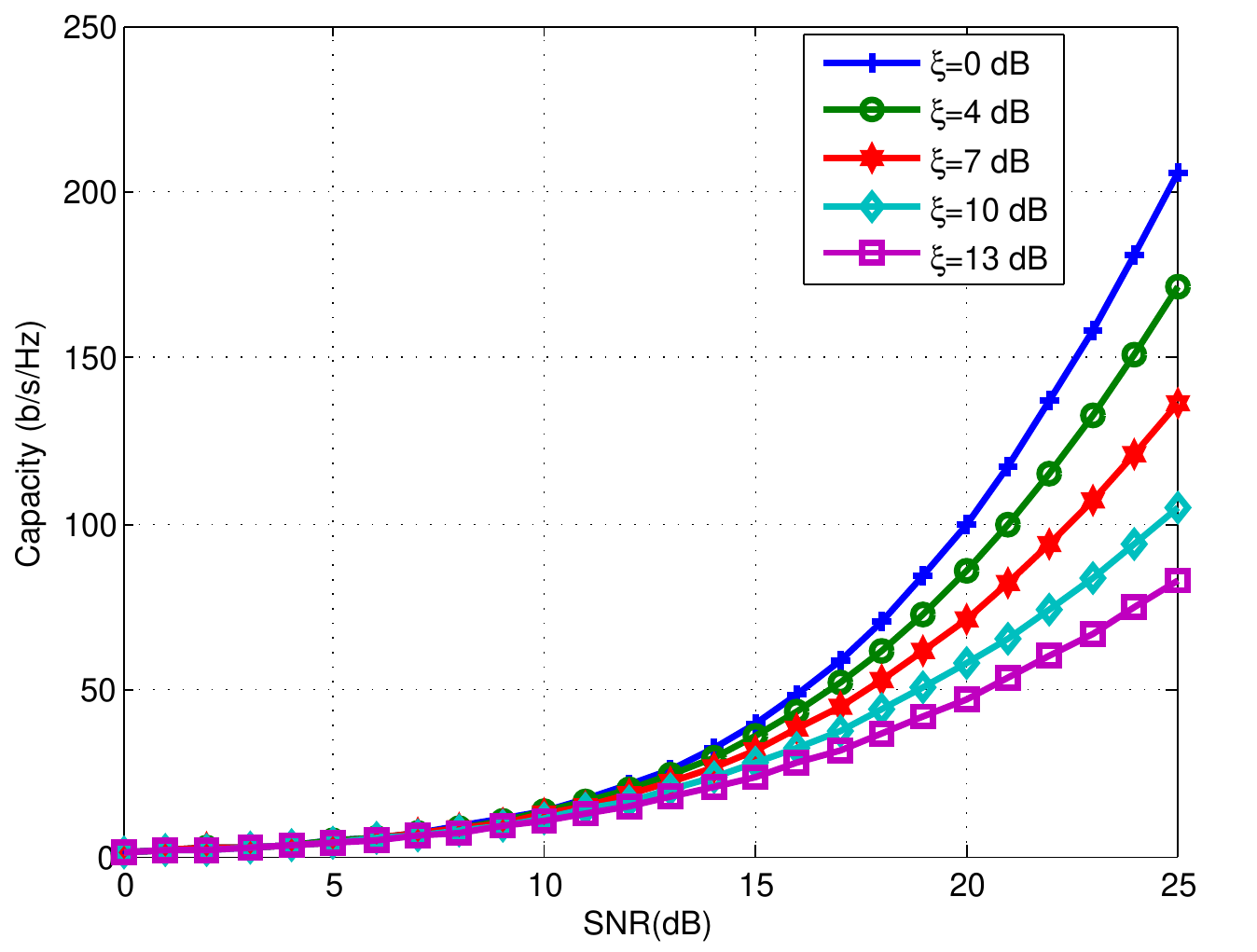}\label{capacity_mdleffect2}}
\caption{Ergodic Capacity Anlysis}
}
\end{figure}

\subsection{Frequency Selective Channel Capacity}
\label{frequency_selective_capacity}
When chromatic and intermodal dispersion are taken into account, the fiber's frequency response $\mathbf{H}\left(\omega\right)$ becomes frequency selective. Under the same linear assumptions as in the previous section, the input-output relationship is given by
\begin{equation}
\label{freq_selective_time_model}
\mathbf{y}\left(t\right)=\mathbf{H}\left(t\right)\star\mathbf{x}\left(t\right)+\mathbf{v}\left(t\right)
\end{equation}
where $\mathbf{v}\left(t\right)$ is a vector Gaussian process, $\mathbf{x}\left(t\right)$ is the input, and $\mathbf{y}\left(t\right)$ is the received signal. Recall, from Section \ref{prop_model}, that $\mathbf{H}\left(\omega_0\right)$ can be written as
\begin{equation}
\label{frequency_flat_model3}
\mathbf{H}\left(\omega_0\right)=\mathbf{U}_{H}\left(\omega_0\right)\mathbf{\Lambda}_{H}\left(\omega_0\right)\mathbf{V}_{H}^{*}\left(\omega_0\right)
\end{equation}
where $\mathbf{U}_{H}\left(\omega_0\right)$, $\mathbf{\Lambda}_{H}\left(\omega_0\right)$, and $\mathbf{V}_{H}^{*}\left(\omega_0\right)$ all depend on $\omega_0$. From \cite{tse_viswanath_wireless_2005}, the capacity of a single instantiation of $\mathbf{H}\left(\omega\right)$, when CSI is not available at the transmitter, is equal to
\begin{equation}
C=\frac{1}{2\pi W}\int_{0}^{2\pi W}\log \det\left(\mathbf{I_{N_r}}+\frac{\SNR}{M}\mathbf{H}\left(\omega\right)\mathbf{H^{*}}\left(\omega\right)\right)d\omega~~~{\tt b/s/Hz}
\end{equation}
where $W$ is the bandwidth of the system in Hz and $\SNR=P/N_0W$ \cite{tse_viswanath_wireless_2005}. This capacity can be achieved by Orthogonal Frequency Division Multiplexing (OFDM) with $N$ sub-carriers (as $N$ tends to infinity). MIMO OFDM modulation is a popular modulation scheme in wireless communications and is currently being developed for the next generation optical systems \cite{Shieh08_OpticalOFDM}. The maximum achievable capacity of a MIMO-OFDM system with $N$ sub-carriers is
\begin{eqnarray}
C&=&\frac{1}{N}\sum_{i=1}^{N}\log \det\left(\mathbf{I_{N_r}}+\frac{\SNR}{M}\mathbf{H}_ {i}\mathbf{H^{*}}_{i}\right) \nonumber \\
&=&\frac{1}{N}\sum_{i=1}^{N}\sum_{n=1}^{M}\log\left(1+\frac{\SNR}{M}\lambda_{n,i}^2\right)~~~{\tt b/s/Hz}\end{eqnarray}
where $\mathbf{H}_{i}=\mathbf{H}\left(\omega_i\right)$ and $\lambda_{n,i}^2$ is the $n^{th}$ eigenvalue of $\mathbf{H}_{i}\mathbf{H^{*}}_{i}$. When CSI is available at the transmitter, waterfilling can be performed to allocate optimal powers across sub-carriers and transmitters.

In the above analysis, we considered the capacity of (\ref{freq_selective_time_model}) for a given instantiation of $\mathbf{H}\left(\omega\right)$. In our work, we focus on analyzing the expected capacity of the frequency selective system which is given by
\begin{equation}
C_{\avg}=\mathbb{E}\left[\frac{1}{N}\sum_{i=1}^{N}\sum_{n=1}^{M}\log\left(1+\frac{\SNR}{M}\tilde{\lambda}_{n,i}^{2}\right)\right]~~~{\tt b/s/Hz}
\end{equation}
To begin with, if we assume that, in each section, all modes experience the same random loss (i.e., the entries of $\mathbf{g}^k$ are perfectly correlated), then $\mathbf{G}^{k}=e^{\frac{1}{2}g^{k}}\mathbf{I_{M}}$. Furthermore, assume that the $K$ sections are statistically identical. Therefore, $\sigma_k=\sigma$ for all $k$ and $\xi^2=K\sigma^2$. In this case, the overall response is given by
\begin{eqnarray}
\label{frequency_selective_flat_model}
\mathbf{H}\left(\omega\right)&=& \mathbf{H}^{K}\left(\omega\right)...\mathbf{H}^{1}\left(\omega\right)\nonumber \\
&=&\mathbf{U}^{K}\mathbf{\Lambda}^{K}\left(\omega\right)\mathbf{V}^{K^{*}}...\mathbf{U}^{1}\mathbf{\Lambda}^{1}\left(\omega\right)\mathbf{V}^{1^{*}} \nonumber \\
&=&e^{\frac{1}{2}\sum_{k=1}^{K}g^{k}}\mathbf{U}^{K}\Theta^{K}\mathbf{T}^{K}\mathbf{A}^{K}\mathbf{V}^{K^{*}}...\mathbf{U}^{1}\Theta^{1}\mathbf{T}^{1}\mathbf{A}^{1}\mathbf{V}^{1^{*}}
\end{eqnarray}
Observe that, even though $\mathbf{H}\left(\omega\right)$ is a function of $\omega$, $\mathbf{H}\left(\omega\right)\mathbf{H}\left(\omega\right)^{*}=e^{\sum_{k=1}^{K}g^{k}}\mathbf{I_M}$ is independent of $\omega$. This means that $\lambda_{n,i}^2=\lambda=e^{\sum_{k=1}^{K}g^{k}}$ is independent of the frequency index $i$ and the mode number $n$. Thus, the average capacity of the fiber is given by
\begin{equation}
C_{\avg}=M\mathbb{E}\left[\log\left(1+\frac{\SNR}{M}\tilde{\lambda}^2\right)\right]
\end{equation}
where the average is taken over the statistics of $\tilde{\lambda}^2$. Observe that the average capacity scales linearly with $M$, the number of modes, and thus the fiber has $M$ degrees of freedom. Therefore, neither group delay nor chromatic dispersion affect the average capacity of the fiber.

We now derive the capacity for the general case (i.e., when the entries of $\mathbf{g}^k$ are potentially independent). The following theorem shows that the statistics of $\mathbf{H}\left(\omega\right)$ are independent of $\omega$.
\begin{theorem}
The statistics of $\mathbf{H}\left(\omega\right)$ are independent of $\omega$.
\end{theorem}
\begin{proof}
Using the same technique as in the proof of Theorem \ref{unitary_property}, we can show that the statistics of $\mathbf{H}^{k}\left(\omega\right)= \mathbf{U}^{k}\Theta^{k}\mathbf{T}^{k}\mathbf{A}^{k}\mathbf{G}^{k}\mathbf{V}^{k^{*}}$ are the same as the statistics of $\mathbf{H}^{k}=\mathbf{U}^{k}\mathbf{G}^{k}\mathbf{V}^{k^{*}}$ by showing that the distribution of $\mathbf{W}^{k}=\mathbf{U}^{k}\Theta^{k}\mathbf{T}^{k}\mathbf{A}^{k}$ is equal to the distribution of $\mathbf{U}^{k}$. Thus, the statistics of $\mathbf{H}^{k}\left(\omega\right)$ are independent of $\omega$.
\end{proof}
This result shows that the statistics of the eigenvalues of $\mathbf{H}_i\mathbf{H}_i^{*}$ are identical for all $i$. Therefore, the average capacity expression can now be rewritten as
\begin{equation}
C_{\avg}=\sum_{n=1}^{M}\mathbb{E}\left[\log\left(1+\frac{\SNR}{M}\tilde{\lambda}_{n}^{2}\right)\right]~~~{\tt b/s/Hz}
\end{equation}
which is identical to the average capacity of frequency flat optical MIMO systems. Therefore, the results of the previous section carry over to the frequency selective case.

\section{Input-Output Coupling Strategies}
\label{ch:coupling}
The capacity analysis presented in Section \ref{ch:capacity} is important,
but it only serves as an upper limit on the achievable rate. This limit can
only be achieved by making use of all available spatial modes. In theory, one
can always design a fiber with a sufficiently small core radius such that a
desired number of modes propagate through the fiber
\cite{Agarwal2002fiber_optics}. In reality, one has to rely on currently
installed optical fibers and available technologies. The state-of-the art OM3
and OM4 MMF technologies have core radii of 50\,$\mu$m with hundreds of
propagation modes. Unfortunately, having a $100 \times 100$ MIMO system is
neither physically nor computationally realizable at the moment. This means
that a more careful look at the effective channel capacity has to be
considered. This is why we now focus on the case when $N_t$ transmit laser
sources and $N_r$ receivers are used. For most of this section, we
assume that intermodal and chromatic dispersions are negligible. Even though
this may seem like a restriction, this assumption serves to simplify
the discussion and presentation of input-output coupling strategies. The
results and procedures we present offer insight and can be extended to the
more general frequency selective case.
\subsection{Input-Output Coupling Model}
The input coupling is described by $\mathbf{C_{I}}$, an $M  \times  N_t$ matrix, and the output coupling is described by $\mathbf{C_{O}}$, an $N_r  \times  M$ matrix. Here, $M$ is much larger than $N_t$ and $N_r$ and the overall response is given by
\begin{eqnarray}
\label{total_response}
\mathbf{H}_{t}&=&\mathbf{C_{O}}\mathbf{H}^{\left(K\right)}...\mathbf{H}^{\left(1\right)}\mathbf{C_{I}} \nonumber \\
&=&\mathbf{C_{O}}\mathbf{H}\mathbf{C_{I}}
\end{eqnarray}
Therefore, for a single instantiation of $\mathbf{H}_t$, the capacity of the channel is given by
\begin{equation}
\label{overall_system_capacity}
C\left(\mathbf{H_t}\right)=\log \det\left(\mathbf{I_{N_t}}+\frac{\SNR}{N_t}\mathbf{H_tH_t^{*}}\right)
\end{equation}
The input-output coupling coefficients (entries of $\mathbf{C_{I}}$ and $\mathbf{C_{O}}$) are complex quantities capturing the effect of both power and phase coupling into and out of the fiber. These coefficients are determined by the system geometry and launch conditions. For example, in order to study the input coupling profile of each light source one needs to specify its exact geometry and launching angle, and then solve the overlap integrals: two dimensional inner products between the laser's spatial patterns and those of each mode
\begin{equation}c_{ij}=\int\int\phi_{i}\left(x,y\right)\phi_{s_j}\left(x,y\right)dxdy\end{equation}
where $c_{ij}$ is the $\left(i,j\right)^{th}$ entry of $\mathbf{C_I}$, $\phi_{i}\left(x,y\right)$ is the $i^{th}$ mode spatial pattern, and $\phi_{s_j}\left(x,y\right)$ is the $j^{th}$ laser source spatial pattern. However, this procedure is cumbersome and offers little insight on the underlying channel physics. In what follows, we provide a simple condition on the input-output couplers. This condition will prove useful when we present an input-output scheme that maximizes the achievable rate of the overall system (Section \ref{optimal}) and impose a statistical model for $\mathbf{C_I}$ and $\mathbf{C_O}$ (Section \ref{random}).
\begin{theorem} If we neglect the power lost due to input coupling inefficiencies, then a necessary and sufficient condition for $\mathbf{C_{I}}$ to be an input coupling matrix is given by
\begin{equation}\label{condition}\left(\mathbf{c_i},\mathbf{c_j}\right)=\delta_{ij}\end{equation} where $\mathbf{c_i}$ represents the $i^{th}$ column of $\mathbf{C_{I}}$ and $\left(\mathbf{a},\mathbf{b}\right)$ denotes the standard Euclidean inner product between the vectors $\mathbf{a}$ and $\mathbf{b}$. This means that the columns of $\mathbf{C_{I}}$ should form a complete orthonormal basis for $\mathbb{C}^{N_t}$. Similarly, if we neglect the power lost due to output coupling inefficiencies, then the rows of $\mathbf{C_{O}}$ should form a complete orthonormal basis for $\mathbb{C}^{N_r}$.\end{theorem}
\begin{proof} Satisfying the energy conservation principle requires that
\begin{equation}
\label{energyconservation}
||\mathbf{C_{I}x}||^2 = ||\mathbf{x}||^2~~ \forall \mathbf{x} \in \mathbb{C}^{N_t}
\end{equation}
This means that the energy of the input vector should be equal to the energy of the mode vector at the input of the fiber. This condition holds whenever the mapping $\mathbf{C_{I}}$ is a linear isometry mapping. In the special case where $\mathbf{C_{I}}$ is a square matrix, a classical result in linear algebra states that $\mathbf{C_{I}}$ has to be a unitary matrix \cite{shilov1971linear}. However, $\mathbf{C_{I}}$ is a tall $M  \times  N_t$ ($M \gg N_t$) rectangular matrix. In this case, the condition in (\ref{energyconservation}) can be rewritten as
\begin{eqnarray}
\label{condition}
 \left(\mathbf{C_{I}x},\mathbf{C_{I}x}\right)&=&\left(\mathbf{x},\mathbf{x}\right)~~ \forall \mathbf{x} \in \mathbb{C}^{N_t}
 %&=&\left(\mathbf{x},\mathbf{C_{I}^{*}C_{I}x}\right)~~ \forall \mathbf{x} \\ \nonumber
 \end{eqnarray}
 or equivalently as
 \begin{equation}
 \label{condition2}
 \left(\mathbf{x},\mathbf{\left[C_{I}^{*}C_{I}-I_{N_t}\right]x}\right)=0~~ \forall \mathbf{x} \in \mathbb{C}^{N_t}
\end{equation}
If $\mathbf{C_{I}^{*}C_{I}=I_{N_t}}$, the condition in (\ref{condition2})
holds and $\mathbf{C_{I}}$ preserves the norm. This choice ensures that the
columns of $\mathbf{C_{I}}$ form a complete orthonormal basis for
$\mathbb{C}^{N_t}$. However, this only proves the sufficiency part of the
theorem. To prove the necessity part, we consider
$\mathbf{B}=\mathbf{C_{I}^{*}C_{I}-I_{N_t}}$ and show that if
(\ref{condition2}) holds, then it is equal to zero. It can be easily verified
that if $\mathbf{B}$ is a diagonal matrix, then
$\left(\mathbf{x},\mathbf{Bx}\right)=0 ~~ \forall \mathbf{x} \in
\mathbb{C}^{N_t}$ implies that $\mathbf{B = 0}$. The same observation holds
if $\mathbf{B}$ is diagonalizable. In this case, one can choose an
orthonormal basis of eigenvectors and map
$\left(\mathbf{x},\mathbf{Bx}\right)=0 ~~ \forall \mathbf{x} \in
\mathbb{C}^{N_t}$ to $\left(\mathbf{\tilde{x}},\mathbf{D\tilde{x}}\right)=0
~~ \forall \tilde{\mathbf{x}} \in \mathbb{C}^{N_t}$ where $\mathbf{D}$ is a
diagonal matrix. In our case, $\mathbf{B}$ is Hermitian and hence it is
unitarily diagonalizable so that $\mathbf{B = 0}$ or alternatively,
$\mathbf{C_{I}^{*}C_{I}=I_{N_t}}$ as desired. A similar proof can be carried
out to show that $\mathbf{C_{O}C_{O}^{*}=I_{N_r}}$.
\end{proof} We note that even though $\mathbf{C_{I}^{*}C_{I}=I_{N_t}}$, $\mathbf{C_{I}C_{I}^{*}\neq I_{N_t}}$ because $\mathbf{C_{I}}$ is of full column rank, but not of full row rank.
We now use this property to show that $\mathbf{C_{I}}$ and $\mathbf{C_{O}}$ should have a special structure.
\begin{theorem} The input-output coupling matrices $\mathbf{C_{I}}$ and $\mathbf{C_{O}}$ can be expressed as
\begin{equation}
\label{in}
\mathbf{C_{I}=U_{I}\left[
                     \begin{array}{c}
                       \mathbf{I_{N_t}} \\
                       \mathbf{0_{\left(M-N_t\right) \times  N_t}} \\
                     \end{array}
                   \right]
V_{I}^{*}}\end{equation}
\begin{equation}
\label{out}
\mathbf{C_{O}=U_{O}\left[I_{N_r} 0_{N_r \times  \left(M-N_r\right)}\right]V_{O}^{*}}
\end{equation}where $\mathbf{U_{I}}$ and $\mathbf{V_{O}^{*}}$ are $M \times  M$ unitary matrices, $\mathbf{V_{I}^{*}}$ is an $N_t \times  N_t$ unitary matrix, and $\mathbf{U_{O}}$ is an $N_r \times  N_r$ unitary matrix.\end{theorem}
\begin{proof}
By the singular value decomposition (SVD), $\mathbf{C_{I}=U_{I}\Lambda_{I}V_{I}^{*}}$ and $\mathbf{C_{O}=U_{O}\Lambda_{O}V_{O}^{*}}$ \cite{shilov1971linear}. The non-zero singular values of $\mathbf{C_{I}}$ are the square roots of the eigenvalues of $\mathbf{C_{I}^{*}C_{I}}$, and $\mathbf{C_{I}^{*}C_{I}}=\mathbf{I_{N_t}}$. A similar argument holds for $\mathbf{C_{O}}$.
\end{proof}
\subsection{Input-Output Coupling Strategies}
\label{optimal}
In this section, we assume that $N_t=N_r$. When CSI is available at the transmitter and the design of $\mathbf{C_{I}}$ and $\mathbf{C_{O}}$ is affordable, a desirable choice for the input-output couplers is the one that maximizes the system's capacity
\begin{eqnarray}
\label{max_capacity}
\left(\mathbf{C^{opt}_{I}},\mathbf{C^{opt}_{O}}\right)&=&\arg \max_{\left(\mathbf{C_{I}},\mathbf{C_{O}}\right)} \log \det\left(\mathbf{I_{N_r}}+\frac{\SNR}{N_t}\mathbf{H_tH_t^{*}}\right) \nonumber  \\
&=& \arg \max_{\left(\mathbf{C_{I}},\mathbf{C_{O}}\right)} \sum_{n=1}^{N}\log\left(1+ \frac{\SNR}{N_t}\lambda_n^2\right)
\end{eqnarray}
where the $\lambda_n^2$'s are the eigenvalues of $\mathbf{H_tH_t^{*}}$. We note that $\mathbf{C^{opt}_{I}}$ and $\mathbf{C^{opt}_{O}}$ should have a structure compliant with (\ref{in}) and (\ref{out}), respectively. Instead of solving the above constrained optimization problem, we provide an intuitive choice for $\left(\mathbf{C_{I}},\mathbf{C_{O}}\right)$ and argue that it leads to a maximized overall capacity through simulations.
\begin{prop}
The capacity of the overall system in (\ref{overall_system_capacity}) is independent of the choice of $\mathbf{V_{I}^{*}}$ and $\mathbf{U_{O}}$ from (\ref{in}) and (\ref{out}).
\end{prop}
\begin{proof}
The capacity of the overall system is given by
\begin{eqnarray}
C\left(\mathbf{H_t}\right)&=&\log \det\left(\mathbf{I_{N_t}}+\frac{\SNR}{N_t}\mathbf{H_tH_t^{*}}\right) \nonumber \\
&=&\log \det\left(\mathbf{I_{N_t}}+\frac{\SNR}{N_t}\mathbf{\mathbf{C_{O}}\mathbf{H}\mathbf{C_{I}}\mathbf{C_{I}}^{*}\mathbf{H}^{*}\mathbf{C_{O}^{*}}}\right) \nonumber \\
&=&\log \det\left(\mathbf{I_{N_t}}+\frac{\SNR}{N_t}\mathbf{U_{O}\Lambda_{O}V_{O}^{*}HU_{I}\Lambda_{I}\Lambda_{I}^{*}U_{I}^{*}H^{*}V_{O}\Lambda_{O}^{*}U_{O}^{*}}\right) \nonumber \\
&=&\log \det\left(\mathbf{I_{N_t}}+\frac{\SNR}{N_t}\mathbf{\Lambda_{O}V_{O}^{*}HU_{I}\Lambda_{I}\Lambda_{I}^{*}U_{I}^{*}H^{*}V_{O}\Lambda_{O}^{*}}\right)
\end{eqnarray}
Therefore, the capacity of the overall system is independent of $\mathbf{V_{I}^{*}}$ and $\mathbf{U_{O}}$ and hence, without loss of generality, we will assume that they are both equal to the identity matrix.
\end{proof}
The following input-output coupling scheme is suggested
\begin{equation}
\label{input-choice}
\mathbf{C_{I}}=\mathbf{V_{H}\left[
                     \begin{array}{c}
                       \mathbf{I_{N_t}} \\
                       \mathbf{0_{\left(M-N_t\right) \times  N_t}} \\
                     \end{array}
                   \right]}\end{equation}
\begin{equation}
\label{output-choice}
\mathbf{C_{O}}=\mathbf{[I_{N_r} 0_{\left(M-N_t\right) \times  N_t}]U_{H}^{*}}
\end{equation}
where $\mathbf{V_{H}}$ and $\mathbf{U_{H}}$ have been defined in (\ref{channel_svvd}). Choosing $\mathbf{V_{O}=U_H}$, $\mathbf{U_{O}=I_{N_t}}$, $\mathbf{U_{I}=V_H}$, and $\mathbf{V_{I}=I_{N_t}}$ leads to an overall response given by
\begin{eqnarray}
\label{proof}
\mathbf{H}_{t}&=&\mathbf{\left[I_{N} 0_{\left(M-N\right) \times  N}\right]\Lambda_{H}\left[
                     \begin{array}{c}
                       \mathbf{I_{N}} \\
                       \mathbf{0_{\left(M-N\right) \times  N}} \\
                     \end{array}
                   \right]} \\ \nonumber
&=&\diag\left(\lambda_1,...,\lambda_{N}\right)
\end{eqnarray}
Thus, the overall MIMO channel is transformed into a set of parallel AWGN channels. Moreover, since the SVD in (\ref{channel_svvd}) sorts the singular values in decreasing order, the signal energy has been restricted to the $N_t$ (out of $M$) least lossy end-to-end eigenmodes.
%Let $\mathcal{E}_{Nmax}\left(\mathbf{HH^{*}}\right)$ denote the set of $N$ largest eigenvalues of the matrix $\mathbf{HH^{*}}$.
The capacity achieved by this choice of input-output coupling is
\begin{equation}
\label{boundi}
C\left(\mathbf{H}_{t}\right)=\sum_{n=1}^{N_t}\log\left(1+ \frac{\SNR}{N_t}e^{\rho_n}\right)~~~{\tt b/s/Hz}
\end{equation}
This capacity could be further increased by pre-processing $\mathbf{x}$ via a
diagonal power allocation matrix $\mathbf{K}$ using waterfilling. Our
strategy is intuitive since we only have $N_t$ degrees of freedom so it would
be wise if we use the $N_t$ least lossy end-to-end eigenmodes to transmit. We
note that even though the effective end-to-end fiber response shows that we
have used the $N_t$ best end-to-end eigenmodes only, $N_t$ signals were
coupled to and collected from all the available physical modes at the input
and output of the fiber.
%The special kind of transmit and receive beamforming allows the user to effectively transmit the signal through the modes that experience the least MDLs thus leading to an increase in the overall capacity even when we do not allocate powers optimally.
%We stress that \left(\ref{bound}\right) is an upper bound on capacity since
Nonetheless, achieving (\ref{boundi}) requires, as discussed before, having CSI at transmitter and using adaptive spatial filters which is typically hard to implement. Even though we did not prove that the above strategy is capacity optimal, our simulations section will show that it appears to maximize the capacity of the overall system.

\subsection{Random Input-Output Coupling}
\label{random}
The design of reconfigurable input-output couplers is expensive and assumes the availability of CSI at the transmitter (which is only feasible when the channel is varying slowly).  More importantly, in many cases, the coupling coefficients are affected by continuous vibrations and system disturbances. Thus, full control over $\mathbf{C_{I}}$ and $\mathbf{C_{O}}$ is not always affordable. In this section, we analyze the capacity when the user does not have control over $\mathbf{C_{I}}$ and $\mathbf{C_{O}}$. This will give us better insight on the achievable capacity of MIMO MMF systems.
We model the coupling coefficients as time varying random variables and impose a physically inspired distribution that respects both the fundamental energy preservation constraint and the maximum entropy principle. Even though we focus on describing the statistical model of $\mathbf{C_{I}}$, our discussion applies equally well for $\mathbf{C_{O}}$.

For an $M \times M$ square matrix $\mathbf{A}$, the energy conservation
principle confirms that $\mathbf{A}$ should belong to
$\mathbb{U}\left(M\right)$. It was proven in \cite{Mezzadri2007_randomMatrix}
that since a Haar measure exists over $\mathbb{U}\left(M\right)$, one could
define a uniform distribution over $\mathbb{U}\left(M\right)$. Therefore, we
choose a random setting where the input coupling matrix $\mathbf{C_{I}}$ has
its $N_t$ columns randomly selected from a square matrix $\mathbf{A}$ that is
uniformly distributed over $\mathbb{U}\left(M\right)$. This distribution
ensures that the columns of $\mathbf{C_I}$ form a complete orthonormal basis
for $\mathbb{C}^{N_t}$ and gives equal probability measure for all such
possible vectors. In other words, $\mathbf{C_{I}}$ is uniformly distributed
over a Stiefel manifold $\mathbb{V}_{N_t}\left(\mathbb{C}^M\right)$. Appendix
\ref{appendix_random_matrices} shows how we can generate $\mathbf{C_{I}}$ and
$\mathbf{C_{O}}$ from an $M  \times  M$ matrix with i.i.d. Gaussian entries.
The ergodic capacity of the overall system can now be computed by averaging
over the statistics of the input-output couplers and the statistics of the
fiber response. Similarly, one could also compute the probability of an
outage event by obtaining the cumulative distribution function (CDF) of the
capacity, which now depends on the statistics of $\mathbf{C_{I}}$ and
$\mathbf{C_{O}}$.
\subsection{Discussion}
We have evaluated the capacity of both controlled and uncontrolled MIMO MMF systems. As discussed in Section \ref{optimal}, the controlled case refers to the case when CSI is available at the transmitter side and there is full control over the input-output couplers. The uncontrolled case refers to the random coupling model presented in Section \ref{random}. In our simulations, we numerically computed the ergodic capacity using (\ref{average_capacity}), with $\mathbf{H}$ replaced by $\mathbf{H_t=C_OHC_I}$. The average in (\ref{average_capacity}) is taken over the statistics of the channel and the input-output couplers for the uncontrolled case. For reference, we included plots of the capacity when
 \begin{itemize}
 \item all mode dependent losses are equal to zero and there is no mode
     coupling ($K=1$ and $\xi=0$); hence the channel has unity
     eigenmodes. A fiber with such properties will be referred to as an
     \emph{ideal fiber}.
 \item the fiber core radius is chosen so that only $N_t$ modes can propagate. In this case the input and output coupling matrices are unitary matrices.
 \end{itemize}
In this analysis, we consider $N_t= N_r= 4$, $K = 256$, $\xi = 4$ dB, and $M = 100$.
\begin{figure}[t]
\centering
\includegraphics[scale=0.6]{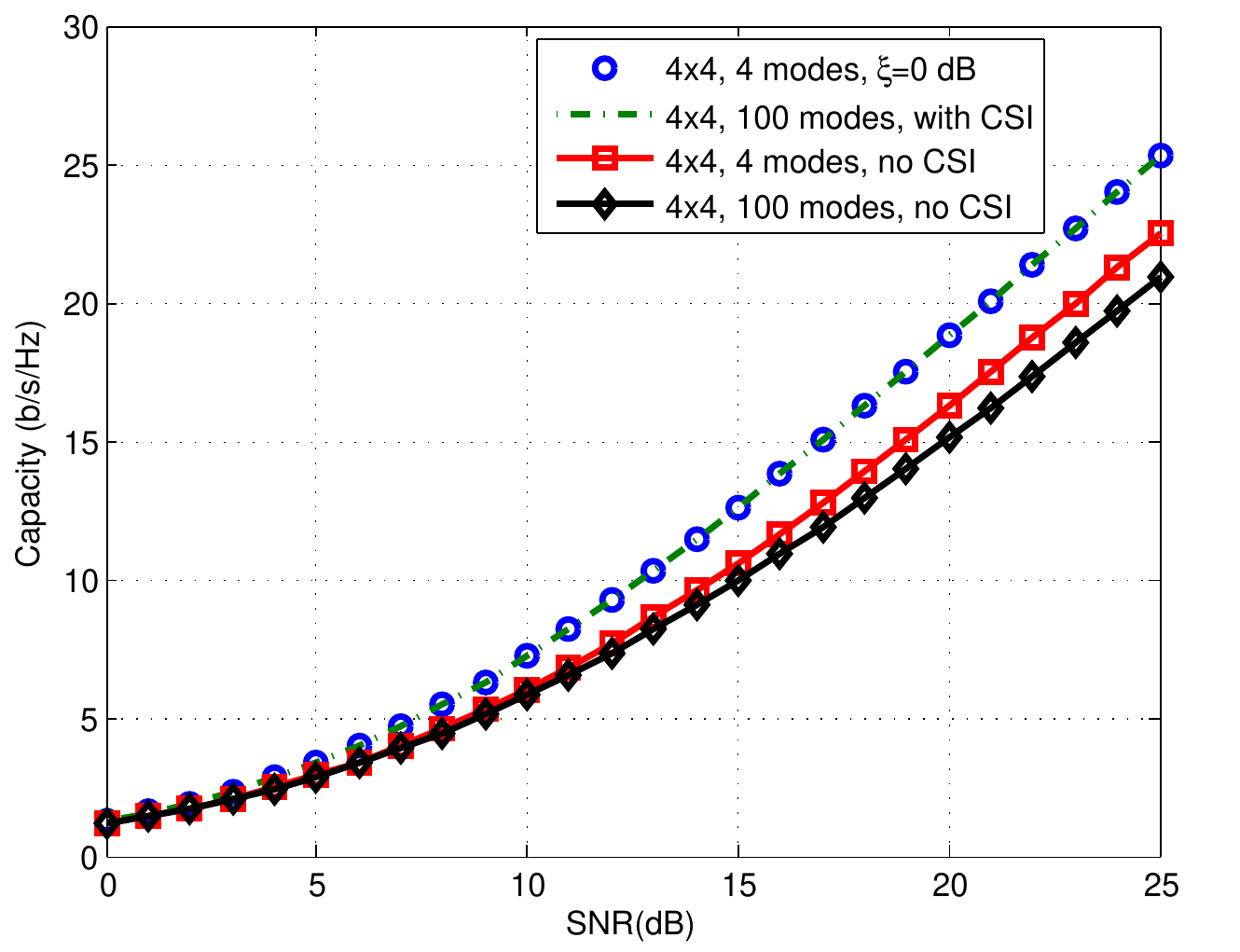}
\caption{Achievable capacity of a $4 \times 4$ MIMO MMF system}
\label{overallcapacity}
\end{figure}
Comparing Figures \ref{overallcapacity} and \ref{capacity_mdleffect2}, we
observe that the capacity of a $4 \times 4$ system over a $100$-mode fiber
is inferior to the intrinsic capacity of the fiber (i.e., when all the modes
are used). This result is expected since we are using $4$ out of $100$
available degrees of freedom. At moderate $\SNR$ values the loss in capacity
is about $6$ dB. On the other hand, observe, from Figure \ref{overallcapacity},
that the performance of an uncontrolled  $4 \times 4$ system over a
$100$-mode fiber is close to that of a system with $4$ modes. Thus,
currently installed fibers could be used without significant loss in
capacity. We also note that, by using the input-output coupling strategy presented in Section \ref{optimal}, performance equal to that of an ideal
fiber can be achieved. This is explained by revisiting Figure
\ref{mdl_distribution} which shows the probability distribution of the
end-to-end MDL values when $M=100$. We observe that in this case we only use
the best $4$ eigenmodes to transmit the signal. As such, it is highly
probable that these $4$ (out of $100$) modes will have close to zero
end-to-end mode-dependent losses, and thus the performance is almost equal to that of an ideal fiber (even when $\xi$ is large). The larger $M$ is, the
closer the capacity of a controlled fiber can get to that of an ideal fiber.
Finally, one could argue that coupling a reasonable number of inputs to a
fiber with hundreds of modes is advantageous since the fiber's peak power
constraint is proportional to the number of modes (recall that more
propagation modes means larger core radius). This means that compared to an
$N_t$-mode fiber, a higher capacity could be achieved if we signal over an $M
\gg N_t$-mode fiber since the total power budget can now be increased.

\section{Conclusion}
\label{ch:conclusion}
MIMO communications over optical fibers is an attractive solution to the ever
increasing demand for Internet bandwidth. We presented a propagation model
that takes input-output coupling into account for MIMO MMF systems. A
coupling strategy was suggested and simulations showed that the capacity of
an $N_t \times  N_t$ MIMO system over a fiber with $M  \gg N_t$ modes can
approach the capacity of an ideal fiber with $N_t$ modes. A random
input-output coupling model was used to describe the behavior of the system
when the design of the input-output couplers is not available. The results
illustrated that, under random coupling, the capacity of an $N_t \times N_t$ MIMO
system over a fiber with $M \gg  N_t$ modes is almost equal to that of an
$N_t$-mode fiber.

%Furthermore, Moore's law is reaching its limits, and device scaling can no
%longer readily provide the increases in electrical switching speed and
%transistor density to which we have become accustomed over the last several
%decades. Optical data rates are sufficiently high that we can no longer
%ignore constraints on receiver computational complexity and devices when
%considering achievable rates.

% APPENDIX
\appendix
\section{Random Unitary Matrices}
\label{appendix_random_matrices}
Our method for generating random unitary matrices is based on the QR
decomposition procedure \cite{Mezzadri2007_randomMatrix}. In this case,
$\mathbf{A}$ is constructed as follows:
\begin{enumerate}
\item  Generate an $M  \times  M$ matrix $\mathbf{Z}$ with i.i.d. complex
    Gaussian entries.
\item  Obtain the QR decomposition of $\mathbf{Z}$; $\mathbf{Z=QR}$.
\item  Form the following diagonal matrix:
\begin{equation}
\mathbf{\Lambda}=\left(
               \begin{array}{cccc}
                 \frac{r_{11}}{|r_{11}|} &  &  &  \\
                  & \frac{r_{22}}{|r_{22}|} &  &  \\
                  &  &  \ddots &  \\
                  &  &  & \frac{r_{MM}}{|r_{MM}|} \\
               \end{array}
             \right)
\end{equation}
where $\{r_{ii}\}_{i=1}^{M}$ are the diagonal entries of $\mathbf{R}$.
\item Let $\mathbf{A}=\mathbf{\Lambda Q}$.
\end{enumerate}
In the above construction, $\mathbf{A}$ is obviously unitary since
$\mathbf{Q}$ is unitary. Furthermore, it can be shown that $\mathbf{A}$ has a
uniform distribution over $\mathbb{U}\left(M\right)$. To generate the input
coupling matrix, the following method is used:
\begin{enumerate}
\item  Generate an $M  \times  M$ unitary matrix $\mathbf{A}$ (as described
    above).
\item  Choose $N_t$ columns randomly from $\mathbf{A}$ to form
    $\mathbf{C_{I}}$.
\end{enumerate}
A similar approach can be taken to generate $\mathbf{C_{O}}$. In this case,
$N_r$ columns are selected randomly from $\mathbf{A}$ to represent the rows
of $\mathbf{C_{O}}$.

\bibliographystyle{IEEEtran}
\bibliography{IEEEfull,mimo_optical}
 %\end{spacing}

\end{document}